\documentclass{article}
\usepackage{graphicx} 
\usepackage[numbers]{natbib}

\usepackage{graphicx}
\usepackage{blindtext}
\usepackage{amsmath}
\usepackage{amsfonts}
\usepackage{amsthm}
\usepackage{amssymb}
\usepackage{tabularx}
\usepackage{mathtools}
\usepackage{epstopdf}
\usepackage{pdfpages}
\usepackage[margin=1in]{caption}
\usepackage{listings}
\usepackage{import}
\usepackage{framed}

\usepackage[utf8]{inputenc}

\newtheorem{lemma}{Lemma}[section]
\newtheorem{theorem}{Theorem}[section]
\newtheorem{corollary}[theorem]{Corollary}

\theoremstyle{remark}

\usepackage[nottoc,notlot,notlof]{tocbibind}
\usepackage{algorithm}
\usepackage{algorithmic}
\usepackage{url}

\title{Generating Networks to Target Assortativity via Archimedean Copula Graphons}
\author{Victory Idowu}
\date{February 2025}

\begin{document}

\maketitle

\begin{abstract} 
We develop an approach to generate random graphs to a target level of assortativity by using copula structures in graphons. Unlike existing random graph generators, we do not use rewiring or binning approaches to generate the desired random graph. Instead, we connect Archimedean bivariate copulas to graphons in order to produce flexible models that can generate random graphs to target assortativity. We propose three models that use the copula distribution function, copula density function and their mixed tensor product to produce networks. We express the assortativity coefficient in terms of homomorphism densities. Establishing this relationship forges a connection between the parameter of the copula and the frequency of subgraphs in the generated network. Therefore, our method attains a desired the subgraph distribution as well as the target assortativity. We establish the homomorphism densities and assortativity coefficient for each of the models. Numerical examples demonstrate the ability of the proposed models to produce graphs with different levels of assortativity. 
\end{abstract}

\textbf{keywords} 
Assortativity, Network theory, Network motifs, Null Model, Copulas. 

Assortativity, is a fundamental property in complex networks which characterizes the structure and dynamics of inter connectivity within a network. Assortativity measures mixing patterns generated by the dependence between a network's nodes \cite{newman2002assortative, newman2003mixing, newman2003mixingb}. Assortative networks are networks where high degree nodes are more likely to connect to other high degree nodes. Whereas disassortative networks are networks where high degree nodes are more likely to connect to other low degree nodes.  In social networks assortativity manifests as homophily \cite{anagnostopoulos2008influence}, where individuals often associate themselves with other individuals which share similar characteristics \cite{zhou2021periodic, aral2009distinguishing}. Conversely, biological \cite{mangioni2018multilayer, yan2017degree, yan2017network} and technological networks exhibit disassortative networks \cite{xiao2010robustness, rubin2016disassortativity}. 

The Newman assortativity coefficient, $r$, introduced in \cite{newman2002assortative} quantitatively measures the extent of assortativity within a network. The coefficient ranges from $r \in [-1,1]$, with assortative mixing indicated by positive values and disassortative mixing indicated by negative values. A value of $0$ indicates no global assortativity nor disassortativity within a network. Statistically, this measure is the Pearson correlation coefficient between the excess degree of the pairs of nodes on all edges in a network.

Generating graphs to target level of assortativity and predefined homomorphism densities is an important task for graph learning algorithms and real-world networks. Graph learning algorithms require graphs to be generated with specific properties benchmarking and evaluating the effectiveness of their algorithms \cite{tsitsulin2022synthetic, lassance2020graph}. For instance, algorithms need to be tested on networks with predefined assortativity \cite{newman2002assortative, tsitsulin2022synthetic} to conduct tasks like node classification \cite{suresh2021breaking} and link prediction \cite{xu2020quantifying, altmann2024synthetic}. Also, algorithms are required to efficiently operate on graphs with specified target motif counts\cite{zhao2012subgraph, rossi2018estimation, ribeiro2021survey} which are critical in tasks like motif discovery and graph matching \cite{ribeiro2021survey}. Current benchmark graph datasets often have predefined properties including assortativity and homomorphism densities \cite{liu2022taxonomy, hu2020open}, but the process behind generation of these graphs are often opaque. 


Generating networks to a target level of assortativity is crucial for studying real-world networks. Many real world networks exhibit proprieties like: differing levels of resilience and robustness \cite{ louzada2013smart, chan2016optimizing, zhao2010achieving, xiao2010robustness}; triangles, clustering and connectivity patterns \cite{foster2011clustering, estrada2011combinatorial, la2018influence, aksoy2019generative}; and influence information diffusion \cite{anagnostopoulos2008influence, xulvi2004reshuffling, aral2009distinguishing, zhou2021periodic}.

Popular network generation methods like the configuration model \cite{bender1978asymptotic, bollobas1980probabilistic, molloy1995critical,molloy1998size}, Chung-Lu model \cite{chung2002connected}, preferential attachment model \cite{barabasi1999emergence} cannot create assortative networks. This inability is due to their generation mechanism. Preferential attachment models cannot create assortative networks as the linking mechanism is proportional to the degree and does not consider degree correlation \cite{dorogovtsev2000structure}. Chung-Lu models can create networks with transitivity coefficients but are often not high enough to reflect real world networks \cite{mussmann2015incorporating}. Configuration models do not induce any degree correlations and produce graphs of no assortativity or disassortativity, which are then rewired \cite{bertotti2019configuration}. Reconstruction of these generative models to develop some assortative configuration model with a block structure \cite{lee2019generalized} considers an assortative configuration model through an assumed block structure. Similarly in \cite{mussmann2014assortativity} binning and accepting-rejection sampling are used to introduce assortativity to the Chung-Lu model. 

Edge rewiring is conducted on networks created from current network generation models in order to generate a network to a target level of assortativity \cite{newman2003mixing, xulvi2004reshuffling, van2010influence, yuan2024strength}. Edge rewiring processes are normally double edge swaps \cite{chan2016optimizing}. Double edge swaps repeatedly change the connected nodes via a Monte Carlo algorithm until the target level of assortativity. An advantage of the double edge swap means that the degree distribution is preserved. Rewiring processes do not consider for the changes in the motifs counts in the process. Motifs, which are subgraphs like the number of triangles, can have a significant effect on a network's topology just like the assortativity coefficient \cite{foster2011clustering, la2018influence}. Rewired graphs however may obtain target levels of assortativity at the expense of real world motif patterns \cite{xiao2010robustness, melamed2018cooperation}.

The aim of this paper is to address the question: Can we generate a random graph to a target level of assortativity, without having to rewire? 

We aim to generate real world networks to capture the presence of triangles and other connectivity patterns \cite{la2018influence}, 
By incorporating the variety of structures into the network generation process, we can better analyze the mixing patterns that characterize social, biological and technical networks. These generated networks in turn would have better predictive power.

In this paper, we propose a method to produce a network to a desired level of the assortativity coefficient without rewiring. By doing so, we make no requirement to initially specify a degree sequence or edge sequence to later rewire to the required assortativity. We also make the following contributions:
\begin{enumerate}
    \item We propose a new graph generation algorithm only requiring a copula structure and number of nodes in the graphon framework.
    \item We define the assortativity coefficient in terms of homomorphism densities of a graph or graphon
    \item We show how to generate graphons using both the copula distribution function, the copula density function, and their tensor product.
    \item We propose an algorithm to generate a network to a target level of assortativity using the copula graphon and copula density graphon; and another algorithm for their tensor product.
\end{enumerate}

This paper is organized as follows. First, we provide a brief overview of graphs, graphon theory, homomorphism density, copulas and assortativity. Next, we show how degree assortativity coefficient can be rewritten in terms of homomorphism densities. Then we provide three methods to generate networks using the copula distribution function, copula density function and their tensor product. Afterwards, we provide illustrative examples that show how to generate a network to a target level of assortativity.

\section{Graphons, Copulas and Assortativity}
This section contains the definition of copulas, graphons and main results of degree assortativity. 

\subsection{Simple graphs and Graph Homomorphisms}
Given a simple graph $G=(V(G),E(G))$, let $i$ be a vertex and $uv = e $ be an edge in $G$. Let $|V(G)| = n$ and $|E(G)| =m$, which is the number of nodes and number of edges respectively. $d_i$ is the degree of vertex $i$. In simple graphs nodes and edges are without labels and undirected. Therefore, the edges $ij$ and $ji$ are the same. $k_i^\prime = d_i - 1$ is the excess degree of node $i$. $p_k$ is the probability that a randomly chosen node has degree $k$. $q_k$ is the distribution of the excess degree, that is, the probability that a randomly chosen edge will have a node of excess degree $k$. It is given by, 
\begin{equation*}
    q_k = \frac{(k+1)p_k}{\sum_k k p_k}
\end{equation*}

The joint degree distribution is $(d_i, d_j)$. The joint excess degree distribution is $(k_i^\prime, k_j^\prime) = (d_i -1, d_j -1)$. The joint excess degree distribution characterizes the degree mixing matrix, $E = \{ e_{ij}\}$, \cite{newman2002assortative} which gives the probability that a randomly chosen edge connects a node of degree $i$ to a node of degree $j$. It's empirical counterpart is denoted $M$.

We now introduce notation for subgraphs $F$ from which we want to count their appearance in a graph $G$. Let $P_i$ be the path on $i$ vertices, $C_i$ be the cycle on $i$ vertices and $S_i$ be the star graph on $i$ vertices with $i+1$ edges. Counts of these subgraphs determine the assortativity coefficient. Note that subgraphs are graphs in their own right. 

Graph homomorphisms are adjacency-preserving maps between graphs. Consider a motif $F$ and a (simple) graph $G$, maps $\alpha: F \rightarrow G$ then $\alpha$ is a mapping. A homomorphism, $\mathrm{hom}(F,G)$ is the total number of graph homomorphisms between $F$ and $G$. Let $V(G) = n$ and $V(F) = m$. An injective homomorphism, $\mathrm{inj}$, is the count of all injective homomorphisms between $F$ and $G$.

The equivalence class of graph homomorphisms is large, as several permutations of labelled graphs lead are created by the same graph homomorphism. We denote the homomorphism count of $F$ in $G$ as $ t \left( F,G \right) $. It will also be referred to as the motif counts. Instead of calculating $ t \left( F,G \right) $, we calculate the homomorphism density between simple graphs \cite{lovasz2006limits}, for $V(F) < V(G)$:
\begin{equation} \label{eq:hom_combin}
    t \left( F,G \right) := \frac{\hom (F,G)}{n^m}
 \end{equation}
as well as the isomorphisms, $t_{\mathrm{inj}}$ or $|F|$, follows:
\begin{equation} \label{eq:inj_combin}
    t_{\mathrm{inj}} \left( F,G \right) := \frac{\mathrm{inj} (F,G)}{(n)_m}
\end{equation}
where $(n)_m$ is the falling factorial. Clearly, these densities are also probabilities and $t(F,G) := \mathbb{P}(F \subseteq G[m])$ where where $G^\prime [m])$ the subgraph induced by picking $m$ nodes of $G$ without replacement, likewise, $t_{\mathrm{inj}} := \mathbb{P}(F \subseteq G^\prime [m])$ where $G^\prime [m])$ the subgraph induced by picking $m$ nodes of $G$ with replacement.

Homomorphism densities as follow identities. For $F_1, F_2$ node-disjoint subgraphs $\hom(F_1 \cup F_2,G)=  \hom \left( F_1,G \right) \hom \left( F_2,G \right)$. 

\subsection{Graphons and their Densities}
The graphon \cite{lovasz2012large, borgs2006counting} is a graph-function that is defined over the space of all bounded symmetric measurable functions, $\mathcal{W}$. A graphon $W \in \mathcal{W}$ is an integrable function $W: [0,1]^2  \rightarrow [0,1]$. The probability of an edge between points $u_i, u_j$ is $W(u_i, u_j)$. $\{ u_i \}$ are a collection of latent random variables and are often modeled as continuous uniforms. Let $\mathbb{G} = \mathbb{G}(n,W)$ be the random graph generated by the graphon $W$ on $n$ nodes, also referred to a $W-$random graph.

Graphons have the advantageous property that network properties can be defined directly from its integral transform. As a graphon is a continuous random variable, it has a degree operator, $ \lambda(u) := \int_0^1 C(u,v) dv$. The degree distribution of a normalized graphon has this distribution function: $D_W(u) := \int_0^u \lambda(t) dt$. 

For graphons network properties can be characterized as "moments", homomorphism densities of the graphon \cite{bickel2011method}. The graphon specifies the limiting sub graph frequencies (probabilities). The homomorphism density of $F_G$ for a given graphon is expressed as: 
\begin{equation}\label{eq:hom-density}
    t(F,W) = \int \prod_{i \in V(F)} dx_i \prod_{ij \in E(F)} W(x_i, x_j)
\end{equation}
Note that $t_{\mathrm{inj}}(F,W) = t(F,W)$ since all maps of $i$ occur with probability 1. We denote the isomorphism density of $W$ by $t_{\mathrm{iso}}(F,W)$. For example the edge density of the graphon is $t(P_1, W) = \int_{[0,1]^2} W(x,y) dx dy$ and the triangle density is given by $t(C_3, W) = \int_{[0,1]^3} W(x,y)W(y,z)W(x,z) dx dy dz$. 

Graphons have the ability to be used as building blocks in advanced network construction \cite{liu2024mixup, han2022g_mixup,saha2024graphon}. One method of construction is to use a tensor product of graphons: 
\begin{equation}
    (U \otimes W)(x_1, x_2, y_1, y_2) := U(x_1, y_1)W(x_2, y_2)
\end{equation}
There exists a measure preserving map $ [0,1] \rightarrow [0,1]^2$. 

The tensor product extends the identity on homomorphism counts as
\begin{equation}\label{c1a-eq-tensorprod}
    t(F, U \otimes W) = t(F,U)t(F,W).
\end{equation}

\subsection{Archimedean copula theory}
Copulas are a popular tool in parametric modeling for discrete and continuous random variables. Copulas have the ability to represent dependence structures in terms of the marginal distribution functions where the joint distribution function is unknown. A $d$-dimensional copula is denoted by $C \left( \boldsymbol{u} \right) = C \left( u_1, \ldots, u_d \right)$ is a symmetric mapping from $C: [0,1]^d \rightarrow [0,1] $ that satisfies the following:
\begin{enumerate}
\item $C \left( u_1, \ldots, u_d \right) $ is increasing in each component $u_i$.
\item $C \left( 1, \ldots, 1, u_i, 1, \ldots, 1 \right) = u_i \quad \forall i \in {1, \ldots, d}, u_i \in \left[ 0,1 \right]$
\item For all $\left( a_1, \ldots, a_d \right), \left( b_1, \ldots, b_d \right) \in \left[ 0,1 \right]^d$ with $a_i \leq b_i$ it holds that:
\end{enumerate}
\begin{equation}\label{eq:rectangleinequality}
\sum_{{i_1} =1 }^2 \ldots \sum_{{i_d} =1 }^2 \left(-1 \right)^{i_1 + \ldots + i_d} C \left( u_{1, {i_1}}, \ldots,u_{j, {i_j}}, \ldots, u_{d, {i_d }} \right)
\end{equation}

where $u_{j 1} = a_j$ and $u_{j 2} = b_j$ for all $j \in \left\{ 1, \ldots, d\right\}$

The existence of a copula for any random vector is guaranteed by Sklar's theorem:
\begin{theorem}
For $x_1, \ldots x_d$ with joint cumulative distribution function $H(x_1, \ldots x_d)$,
\begin{equation*}
    H(x_1, \ldots, x_d) = C(F(x_1), \ldots, F(x_d))
\end{equation*}
\end{theorem}

All copulas are bounded above and below by the Fr\'{e}chet-H\"{o}ffding bounds \cite{nelsen2006introduction}:
\begin{theorem} \label{F-H bounds}
For every bivariate copula, $C \left(u_1, u_2 \right)$ the upper and lower bounds are given by:
\begin{equation}\label{eq:standardfhbounds}
    C^{-}(u_1. u_2) \leq C(u_1. u_2) \leq C^{+}(u_1. u_2)
\end{equation}
where $C^{-} := \max\left\{ u_1 + u_2 -1,0 \right\} $ and $C^{+} := \min\left\{ u_1, u_2\right\}$. We will drop the $u_1, u_2$ when clear. 
\end{theorem}
The limits of the Fr\'{e}chet-H\"{o}ffding bounds are called the countermonotonic (maximum) copula $C^{-}$ and the comonotonic (minimum) copula $C^{+}$. The most basic copula is the independence copula $C(u_1, u_2) = u_1 u_2$. We will denote the independence copula by $\Pi$. $C^{-}$, $C^{+}$ in addition with $\Pi$ are called the \textit{fundamental copulas}.    

The most used copulas are from the Archimedean copula family, also known as explicit copula since they admit an explicit formulae \cite{mcneil2015quantitative}. Archimedean copulas are uniquely determined by their generator $\varphi_\theta$, a function that satisfies: (i) $\varphi_\theta: \left[ 0, 1\right] \rightarrow \left[ 0, \infty \right]$, (ii) $\varphi_\theta$ is a continuous and strictly decreasing function in the interval $\left[0,1 \right]$ and (iii) $\varphi_\theta$ is such that $\varphi_\theta (1) = 0$. The pseudo inverse of the generator  is defined as:
\begin{equation}\label{phi-additive-inverse}
  \varphi_\theta^{[-1]} (x) =     \begin{cases}
                        \varphi_\theta^{-1} (x) & 0 \leq x \leq \varphi_\theta (0) \\
                        0 & \varphi_\theta (0) < x \leq \infty
                        \end{cases}
\end{equation}

(Bivariate) Archimedean copulas are with the explicit form defined by their generator:
\begin{equation} \label{arch-phi-definiton}
  C_{\theta} \left( u_1, u_2 \right) =  \varphi_\theta^{\left[ - 1 \right]} \left( \varphi_\theta( u_1) +\varphi_\theta(u_2) \right) \qquad \text{ for } u_i \in [0,1].
\end{equation}

Assuming $\varphi(t)$ is twice differentiable, the copula density function is \cite{nelsen2006introduction}:
\begin{equation*}
\begin{aligned}
    c_{\theta} \left( u_1, u_2 \right) &= \frac{\varphi^{\prime \prime}(C(u,v)) \varphi^\prime(u) \varphi^\prime(v)}{\left( \varphi^\prime(C(u,v))\right)^3} \\
    & = \frac{\varphi^{\prime \prime}(\varphi_\theta^{\left[ - 1 \right]} \left( \varphi_\theta( u_1) +\varphi_\theta(u_2) \right)) \varphi^\prime(u) \varphi^\prime(v)}{\left( \varphi^\prime(\varphi_\theta^{\left[ - 1 \right]} \left( \varphi_\theta( u_1) +\varphi_\theta(u_2) \right))\right)^3}
\end{aligned} 
\end{equation*}

We drop $\theta$ in $\varphi_{\theta}$ when the copula being described is clear. 

Four popular Archimedean copulas are the Clayton copula $C_C$, Frank copula $C_F$, Gumbel copula $C_G$ and Joe copula $C_J$. For each of these copulas their generator function, bivariate copula and copula density as follows. More Archimedean copulas in \cite{nelsen2006introduction, joe2014dependence}.

The Clayton copula, has $ \theta \in [-1, \infty)$, its generator is 
\begin{equation*}
   \varphi_{C,\theta}(t) = \frac{1}{\theta} (t^{-\theta} - 1), 
\end{equation*}

its copula function is
\begin{equation*}
    C_C(u,v) = \max \left( u^{-\theta} + v^{-\theta} -1 ,0 \right)^{-1 / \theta}.
\end{equation*}

If $u^{-\theta} + v^{-\theta} -1 >0$ the copula density function is,
\begin{equation*}
    c_c(u,v) = (1 + \theta)(uv)^{-\theta - 1} ( u^{-\theta} + v^{-\theta} -1 )^{-2 -1 / \theta}.
\end{equation*}
 
The Frank copula has $ \theta \in (-\infty, \infty) / {0}$, its generator is 
\begin{equation*}
    \varphi_{F,\theta}(t) = -\log \left( \frac{\exp (-\theta t)  - 1}{\exp (-\theta) - 1} \right),
\end{equation*}

its copula function is
\begin{equation*}
    C_F(u,v) = -\frac{1}{\theta} \log \left( (1 + \frac{\exp (-\theta u)  \exp (-\theta v) }{\exp (-\theta)  - 1}\right),
\end{equation*}

its copula density function is
\begin{equation*}
    c_F(u,v) = \frac{\theta  - e^{-\theta}) e^{-\theta(u+v)}}{(1 - e^{-\theta} - (1- e^{-\theta u})(1- e^{-\theta v}) )^2}.
\end{equation*}

The Gumbel copula, $\theta \in [1, \infty)$, its generator is 
\begin{equation*}
    \varphi_{G,\theta}(t) = (- \log(t))^\theta,
\end{equation*}

its copula function is
\begin{equation*}
     C_G(u,v) = \exp \left( - \left[ (-\log(u))^\theta + (-\log(v))^\theta \right]^{\frac{1}{\theta}} \right),
\end{equation*}

and its copula density function is,
\begin{equation*}
\begin{aligned}
     c_G(u,v) &= \exp(-(x^\theta + y^\theta )^{1 / \theta}) \left( (x^\theta + y^\theta )^{1 / \theta} + \theta - 1\right) \cdot \\ &(x^\theta + y^\theta )^{1 / \theta - 2} (xy)^{\theta - 1} (uv)^{-1}   
\end{aligned}
\end{equation*}
where $x = - \log u$ and $y = - \log v$.

The Joe copula has $\theta \in [1, \infty)$ its generator is 
\begin{equation*}
    \varphi_\theta(t) = -\log \left( 1 - (1-t)^\theta \right),
\end{equation*}

its copula function is
\begin{equation*}
    C(u,v) = 1 - \left[ (1-u)^\theta +  (1-v)^\theta - (1-u)^\theta (1-v)^\theta \right]^{\frac{1}{\theta}},    
\end{equation*}

and its copula density function is,
\begin{equation*}
\begin{aligned}
    c_J(u,v) & = (\Bar{u}^\theta +\Bar{v}^\theta - \Bar{u}^\theta \Bar{v}^\theta)^{-2 + 1 / \theta} \; \Bar{u}^{\theta -1} \Bar{v}^{\theta -1} \cdot \\    
    & \left[ \theta - 1 + \Bar{u}^{\theta} + \Bar{v}^{\theta} ] - \Bar{u}^{\theta}\Bar{v}^{\theta} \right]
\end{aligned}
\end{equation*}
where $\Bar{u} = 1 -u$ and $\Bar{v} = 1 -v$.

\subsection{Degree assortativity}
Degree assortativity is a property in networks in which degree correlations govern the tendency for networks to connect. A network is termed assortative, when high degree nodes are more likely to connect to other high degree nodes. A network is disassortative, is when high degree nodes are more likely to connect to low degree nodes. A network can be neither assortative or disassortative. 

Newman's (degree) assortativity coefficient, is a measure of the extent of assortativity across all the nodes of a network $r_S$, \cite{newman2002assortative, newman2003mixing}:
\begin{equation*}
    r_S = \frac{1}{\sigma_q^2} \sum_{jk} (e_{ij} - q_j q_k)
\end{equation*}
where $q_i$ is the normalized excess degree distribution, $e_{jk}$ is the joint excess degree probability distribution and $\sigma_q^2$ the variance of this distribution. As Newman's degree assortativity coefficient  is a network analogue of the Pearson correlation coefficient  $-1 \leq r \leq 1$. 

Clearly, assortative networks have $ r>0$, disassortative networks have $ r<0$, and networks which are neither assortative or disassortative are $r=0$. Perfectly disassortative networks are closer to a randomly generated network, around \cite{newman2003mixingb}, with $r \approx 0$. 

The empirical version of $r$ is calculated from $M$ 
\begin{equation}
    r = \frac{m^{-1} \sum_i j_i k_i - \left[ m^{-1} \sum_i \frac{1}{2} (j_i + k_i  )\right]^2  }{m^{-1} \sum_i \frac{1}{2} (j_i^2 + k_i^2 )  - \left[ m^{-1} \sum_i \frac{1}{2} (j_i + k_i  )\right]^2}
\end{equation}
where $j_i, k_i$ are the degrees of vertices at the ends of the $i$th edge for $G= G(n,m)$. 

Our contribution enables a network to be generated to a target assortativity while determining the desired motif counts.

\section{Networks Generated to Target Assortativity}

This section contains our main results on the properties of graphs generated by graphons to target assortativity. First, we show how the assortativity coefficient can be rewritten in terms of graphon subgraph counts. Then, we extend this approach to the homomorphism densities of an arbitrary graph. Finally, we consider the properties of the assortativity coefficient defined under subgraph counts. 

The degree assortativity coefficient Newman can be written as a combinatorial ratio \cite{estrada2011combinatorial, vasques2020transitivity} of isomorphism counts for a given graph $G$:
\begin{equation}
    \label{eqn:thm-assortativity}
    r = \frac{|P_2| (|P_{3 / 2}| + C - |P_{2 / 1}|)}{ 3|S_{3}| - |P_2|(|P_{2 / 1}| - 1) }
\end{equation}
where $P_i$ is the path on $i$ nodes with $i-1$, $S_i$ star graphs on $i$ nodes  and $i+1$ edges, $C = \frac{3 |C_3|}{P_2}$ is the clustering coefficient and $|P_{r / s}| := |P_r| / |P_s|$. 

\begin{figure}[!t]
\centering
\includegraphics[width=2.5in]{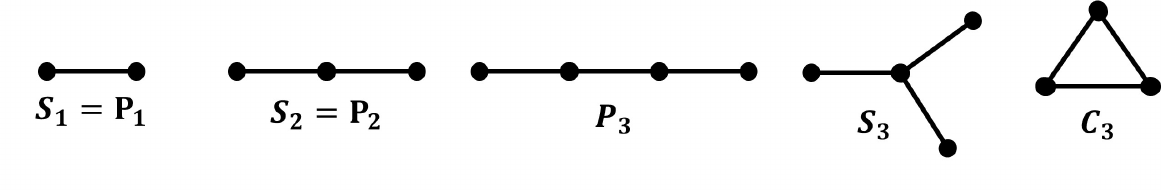}
\caption{Star and path networks in the combinatorial expression of Newman's assortativity coefficient}
\label{c1a-stargraph-figs}
\end{figure}

We note that $r$ is determined by the network's degree and four subgraph counts, $S_1$, $S_2$, $S_3$ and $C_3$ (Fig \ref{c1a-stargraph-figs}). The sign of $r$ is determined by $|P_{3 / 2}| + C - |P_{2 / 1}|$ since the denominator of $r$ is the variance of the degrees \cite{estrada2011combinatorial}. Conditions for degree assortativity sign follows \cite{estrada2011combinatorial, allen2017two}:
\begin{enumerate}
    \item assortative with $r >0$, if and only if $|P_{3/2}| + C > |P_{2/1}|$
    \item $r=0$  if and only if $|P_{3/2}| + C = |P_{2/1}|$ and $3|S_{3}| - |P_2|(|P_{2 / 1}| - 1) \neq 0$
    \item disassortative with $r<0$, if and only if $|P_{3/2}| + C < |P_{2/1}|$.
\end{enumerate}

The homomorphism counts and injective counts for the functionals in $r$ can be calculated analytically using the degrees and edges of a network, as shown in \cite{alon1997finding, estrada2011combinatorial}. It follows that $|P_1| = \frac{1}{2} \sum_{i=1}^n d_i$, $|P_2| =\sum_{i=1}^n {d_i \choose 2} $, $|C_3| = \frac{1}{6}tr(\mathbb{A}^3)$, $|P_3| = \sum_{ij \in E} (d_i - 1)(d_j-1) - 3|C_3|$ and $|S_3| = \sum_{i=1}^n {d_i \choose 3}$.

By using (\ref{eq:hom-density}), we can relate the degree-degree assortativity with the graphon function.
\begin{theorem}[Degree Assortativity Coefficient with Graphon Homomorphism Densities]\label{c1a-thm-graphonr}
Let $W \in \mathcal{W}$ be a graphon on $n$ nodes. The degree-degree assortativity coefficient characterized by homomorphism counts is:
\begin{equation}
    r_W = \frac{(n-3) t(P_3,W) + \frac{3n t(C_3, W)}{(n-1)(n-2)} - \frac{(n-2) t (P_2,W)^2}{t(P_1,W)}}{(n-3)3t(S_3,W) + t(P_2,W) -  \frac{(n-2) t(P_2,W)^2}{t(P_1,W)} }.
\end{equation}
\end{theorem}

\begin{proof}
From (\ref{eq:hom_combin}) (\ref{eq:inj_combin}), we can write both the injective homomorphism and homomorphism counts in terms of their densities. Note, the $|C_3|$ is a homomorphism count and all $|P_i|, |S_3|$, $i \in \{ 0,1,2,3\}$ are injective homomorphism counts. Thus $|P_1| = (n)_2 = t_\mathrm{inj}(P_1, W)$, $|P_2| = (n)_3 t_\mathrm{inj}(P_2,W)$, $|P_3| = (n)_4 t_\mathrm{inj}(P_3,W)$, $|S_3| = (n)_4 t_\mathrm{inj}(S_3,W)$ and $|C_3| = n^3 t(C_3,W)$. For graphons $t_\mathrm{inj} = t$ from which the result follows.
\end{proof}
Note that $r_W$ is defined for the complete graph $K_n$ where all nodes have the same degree as $\sigma_q =0$. Hence $r_W$ is defined for $\sigma_q >0$. 

If the theoretical homomorphism densities from a simple graph $G$ is known, its degree assortativity is as follows.
\begin{corollary}
For a simple graph $G$, its degree assortativity coefficient is:
\begin{equation}
    r_G = \frac{(n-3) t_\mathrm{inj}(P_3,G) + \frac{3n t(C_3, G)}{(n-1)(n-2)} - \frac{(n-2) t_\mathrm{inj}(P_2,G)^2}{t_\mathrm{inj}(P_1,G)}}{(n-3)3t_\mathrm{inj}(S_3,G) + t_\mathrm{inj}(P_2,G) -  \frac{(n-2) t_\mathrm{inj}(P_2,G)^2}{t_\mathrm{inj}(P_1,G)} }
\end{equation}
\end{corollary}
Indeed, this is immediate as that for graph $F$ and simple $G$ $t_{\mathrm{inj}}(F,G) \neq t(F,G)$.

Next, we establish that graphs sampled from graphons have assortativity coefficients that converge to true $r$. Graphs sampled from graphons will be sampled at different resolutions and have fluctuations in $r$ from their true value based on their theoretical homomorphism densities. Hence, the following theorem establishes that a sequence of sampled graphs from $W$ will converge to $r$. 

Let $r_{\mathbb{G}}$ be the empirical assortativity coefficient of $\mathbb{G}(n,W)$ for a graphon $W$ calculated using (\ref{eqn:thm-assortativity}). 

\begin{theorem}[Convergence to the Homomorphism Density Degree Assortativity Coefficient]\label{c1a-thm-assortcoeff}
Let $r_{\mathbb{G}}$ be the assortativity coefficient of $\mathbb{G}(n,W)$  As $n \rightarrow \infty$, $r_{\mathbb{G}} \rightarrow r_W$.
\end{theorem}

The choice of generation of $W, \theta$ will create different convergence speeds of $r_{\mathbb{G}(n,W)}$ to $r_W$.  The proof of this result requires the following lemma on the convergence of homomorphism densities of $W-$random graphs. 
\begin{lemma}(Lemma 2.4 in \cite{lovasz2006limits}) \label{thm-lovaszszeg-hom-1}
Let $F$ be a graph with and $G(n,W)$ be the $W-$random graph of $W$ on $n$ nodes. Then
\begin{enumerate}
    \item $\mathbb{E} (t_{\mathrm{inj}}(F,G(n,W))) = t(F,W)$ 
    \item $| \mathbb{E}(t(F,G(n,W)) - t(F,W)| < \frac{1}{n} { |V(F)| \choose 2}$
    \item $\mathrm{Var}(t(F,G(n,W))) \leq \frac{3}{n} |V(F)|^2$
\end{enumerate}
\end{lemma}
The proof of Theorem \ref{c1a-thm-assortcoeff} is as follows, 
\begin{proof}
From Theorem \ref{thm-lovaszszeg-hom-1}, it follows that, for $i \geq 1$
\begin{equation}
    \mathbb{E}(t_{\mathrm{inj}}(P_i,G(n,W)))= t(P_i,W)
\end{equation}
and
\begin{equation}
    \mathbb{E}(t_{\mathrm{inj}}(S_i,G(n,W)))= t(S_i,W).
\end{equation}
Then,
\begin{equation}
    | \mathbb{E}(t(C_3,G(n,W)) - t(C_3,W)| < \frac{3}{n}
\end{equation}

As $\sigma_q^2 >0$, by Slutsky's theorem, convergence of $r^\star \rightarrow r$.
\end{proof}

\section{Theoretical Results}

\subsection{Copula Graphon} \label{ssec:copulagraphon}
In this section we show how Archimedean copulas can be used as a graphons and their resultant properties. 

Graphons can be represented by any symmetric measurable function, in which bivariate cumulative distributions are a subset, e.g. $W(x,y) = xy$ \cite{han2022g_mixup}. By Sklar's theorem any bivariate distribution function has a copula representation and therefore could be a graphon. 

Let $G$ be a simple graph with graphon $W$ with assumed copula structure $C$. The adjacency matrix $A$ of $G$ is as follows:
        \begin{equation}
            A_{ij}| \mathbf{u} \sim Bernoulli \left( C(u_i, u_j)\right) 
        \end{equation}
        where $u_i, u_j$ are uniform i.i.d random variables.
Copula graphons generated from the four Archimedean copulas: Clayton, Frank, Gumbel and Joe are referred to as $W_C, W_F, W_G$ and $W_J$ respectively. 

Dense graphs can be parameterized by evaluating different moments of the graphon and equating it to a similar definition in copula theory. For instance, direct method of moments estimators based on $\tau$ is possible. Note that this is the copula cumulative density function and not its density. 

Additional properties of the network can be controlled for by comparing homomorphism densities to their integral expressions under the Archimedean copula. These have the added advantage a general form and can be explicitly expressed in terms of the generator function. 

\begin{theorem}\label{thm:C1a-edgeden}[Edge Density of the Copula Graphon] Under the copula graphon framework for dense graphs, the edge density, $t(P_1,W)$ is 
\begin{equation}\label{eqc1a:arch-dense-edge-dens}
    \frac{1}{2} - \int_1^{\varphi^{-1}(1)} \int_1^{\varphi^{-1}(1)}   \varphi(x) \; \varphi^{\prime}(x +y) \; \varphi^{\prime}(y) \; dx \; dy.
\end{equation}
\end{theorem}

\begin{proof}
For Archimedean copulas the explicit form of $ C_{u_2|u_1} \left( u_2 | u_1 \right) $ is known \cite[pp. 91]{joe2014dependence}
    \begin{equation}
        C_{u_2|u_1} \left( u_2 | u_1 \right) = \frac{\varphi^\prime \left(  \varphi^{-1}(u_1) + \varphi^{-1}(u_2) \right) }{\varphi^\prime \left( \varphi^{-1}(u_1) \right) }
    \end{equation}
and  $\varphi^\prime (t) = \frac{d \varphi (t)}{dt}$.   
    It follows,
    \begin{equation}\label{lap-edgedensity-2}
        \frac{1}{2} - \int_0^1 \int_0^1  u_1  \frac{\varphi^\prime \left(  \varphi^{-1}(u_1) + \varphi^{-1}(u_2) \right) }{\varphi^\prime \left( \varphi^{-1}(u_1) \right) }  du_1 du_2.
    \end{equation}

    Note that in (\ref{lap-edgedensity-2}), 
     \begin{equation}
           \begin{aligned}
           \left( u_1 \varphi( \varphi^{-1}(u_1) + \varphi^{-1}(u_2) ) \right)' &= u_1  \frac{\varphi^\prime \left(  \varphi^{-1}(u_1) + \varphi^{-1}(u_2) \right) }{\varphi^\prime \left( \varphi^{-1}(u_1) \right) }  \\ 
           & + \varphi( \varphi^{-1}(u_1) + \varphi^{-1}(u_2)) 
          \end{aligned}
     \end{equation}
     
     (\ref{lap-edgedensity-2}) becomes
     \begin{eqnarray*}\label{lap-edgedensity-3}
     \frac{1}{2} -   \int_0^1 \varphi(\varphi^{-1}( 1 )+ \varphi^{-1}(u_2)) du_2 + \\ \int_0^1 \int_0^1 \varphi( \varphi^{-1}(u_1)+ \varphi^{-1}(u_2)) du_1 du_2
     \end{eqnarray*}
     
\end{proof}

\begin{theorem}[Degree operator of the Copula Graphon] Under the copula graphon framework for dense graphs, the edge density is degree operator becomes: 
\begin{equation}\label{eqc1a:arch-dense-degree-operator}
    \lambda(x) = -\int_x^{\infty} \frac{\varphi^{-1}(s)}{\varphi^{\prime}(\varphi^{-1}(s-x))} ds
\end{equation}
where $x = \varphi(u_1)$.
Moreover, $\lambda(x)$ is bounded above by $\varphi^{-1}(x)$.
\end{theorem}

\begin{proof}
The degree operator is given by,
\begin{equation}\label{eq:do1}
    \int_0^1 \varphi^{[-1]}(\varphi(u_1) + \varphi(u_2)) du_2
\end{equation}

Call $\varphi(u_1) = x$ and take the substitution  $\varphi(u_2) = y$. It follows that (\ref{eq:do1}) becomes:
\begin{equation}\label{eq:do2}
     -\int_0^\infty \frac{\varphi^{-1}(x+y)}{\varphi^{\prime} \left(\varphi^{-1}(y)\right)}  dy
\end{equation}
Taking another substitution, where $s = y-x$, the result follows.
\end{proof}

Using a single copula graphon to generate a dense network produces a limited range of network structures due to the Fr\'{e}chet-H\"{o}ffding bounds. 
\begin{theorem}[Average degree of Archimedean copula graphons] Let $G$ be a simple graph with graphon $W$ with assumed copula structure $C$. Then,
\begin{equation} \label{eqc1a:dense-edge-dens}
    \frac{(n-1)}{6} \leq \mathbb{E}(d_i) \leq \frac{(n-1)}{3}
\end{equation}
\end{theorem}

\begin{proof}
This result is a direct consequence of the Fr\'{e}chet-H\"{o}ffding bounds. For the lower bound $\mathbb{E}(A_{ij}) = \int_0^1 \int_0^1 \max(u+v-1,0) du dv = \frac{1}{6}.$  Similarly for the upper bound $\mathbb{E}(A_{ij}) = \int_0^1 \int_0^1 \min(u,v) du dv = \frac{1}{3}$. From which the result is immediate.
\end{proof}

\begin{theorem}
The star density is $t(S_k, W)$
    \begin{equation} \label{c1a-eq-stardensity}
    \int_0^1 \lambda(x)^k dx = - \int_0^1 \left( \int_x^{\infty} \frac{\varphi^{-1}(s)}{\varphi^{\prime}(\varphi^{-1}(s-x))} ds \right)^k dx
    \end{equation}
\end{theorem}

The integral for the cycle count 
\begin{eqnarray*}
    t(C_k,W) = & \int_{[0,1]^k} \varphi^{[-1]}( \varphi(u_1) + \varphi(u_2)) \ldots \\ &\varphi^{[-1]}(\varphi(u_{k-1}) + \varphi(u_k)) d u_1 \ldots d u_k
\end{eqnarray*}
and for $k \geq 1$,
\begin{eqnarray*}
   &  t(P_k,W) =  \int_{[0,1]^k}  \varphi^{[-1]}( \varphi(u_1) + \varphi(u_2))  \varphi^{[-1]}( \varphi(u_2) + \varphi(u_3)) \\ &  \ldots \varphi^{[-1]}( \varphi(u_{k-1}) + \varphi(u_k)) d u_1 \ldots d u_k
\end{eqnarray*}
can be simplified for specific copula graphons. 

From the integrals (\ref{eqc1a:arch-dense-degree-operator}), (\ref{eqc1a:dense-edge-dens}) and (\ref{c1a-eq-stardensity}), $r$ can be estimated. Clearly, network functionals are determined by the shape of the Archimedean generator function, which is characterized by $\theta$. 

\begin{table}
    \renewcommand{\arraystretch}{1.3}
    \centering
    \caption{Homomorphism densities for the $P_1, P_2, P_3, C_3$ and $S_3$ subgraphs for the $C^{-}, \Pi, C^{+}$ copula graphons}
    \label{c1a-homodensity-tab}
        \begin{tabular}{c|c|c|c|c|c}
        \hline
        $W$ & $t(P_1,W)$ & $t(P_2,W)$ & $t(P_3,W)$ & $t(C_3,W)$ & $t(S_3,W)$\\ \hline \hline
        $C^{-}$& 1/6 & 1/20 & 1/280 & 1/120 & 1/56 \\
            $\Pi$ & 1/4 & 1/12 & 1/36 & 1/27 & 1/32 \\
        $C^{+}$ & 1/3 & 2/15 & 1/14 & 2/15 & 2/35\\
        \hline
    \end{tabular}    
\end{table}

Homomorphism densities for copula graphons are the Fr\'{e}chet-H\"{o}ffding bounds are given in Table \ref{c1a-homodensity-tab}. From which these probabilities we can approximate $r$ for $1000$ nodes as, $C^{-}$ has $r = -0.3$, $C^{+}$ has $r = 0.1$ and $\Pi$ creates $r=0$ neither assortative or disassortative networks.

\subsection{Copula Density Graphon}\label{sec:c1a-copuladensgraphon}
Copula density functions can be used to generate networks with stronger dependence patterns. We denote any copula density graphon with $\Tilde{W}$. Immediately, the Clayton, Frank, Gumbel, and Joe copula density functions are $\Tilde{W}_C$, $\Tilde{W}_F$, $\Tilde{W}_G$ and $\Tilde{W}_J$ respectively. The process to generate copula graphon densities is the same as the copula graphon with a change of function.

The copula density graphon has a similar generative model to the copula graphon. Let $G$ be a simple graph with graphon $\Tilde{W}$ with assumed copula density function $c$. The adjacency matrix $A$ of $G$ is as follows:
        \begin{equation}
            A_{ij}| \mathbf{u} \sim Bernoulli \left( c(u_i, u_j)\right) 
        \end{equation}
        where $u_i, u_j$ are uniform i.i.d random variables.

Copula density functions are probability density functions, hence $t(P_1,\Tilde{W}) =1$. 
\begin{theorem}[Degree operator of the Copula Density Graphon]
For a copula density graphon $\Tilde{W}$ with generator function $\varphi$, its degree operator has closed form:
\begin{equation}
    \lambda(x) = \frac{\varphi^\prime ( \varphi^{[-1]}(x) +\varphi^{[-1]}(1)) }{\varphi^\prime (\varphi^{[-1]}(1)  ) } - \frac{\varphi^\prime ( \varphi^{[-1]}(x) +1) }{\varphi^\prime (1  ) }
\end{equation}
\end{theorem}
\begin{proof}
    Clearly, the degree operator is equivalent to the first partial derivative of the copula function, $u_2 \in \{ 0,1\}$. 
\end{proof}

From which we can deduce $t(S_k, \Tilde{W})$, 
\begin{theorem}[Star density of Copula Density Graphon]
For a copula density graphon $\Tilde{W}$, $t(S_k,W)$ 
\begin{equation}
    \int_0^1 \left( \frac{\varphi^\prime ( \varphi^{[-1]}(x) +\varphi^{[-1]}(1)) }{\varphi^\prime (\varphi^{[-1]}(1)  ) } - \frac{\varphi^\prime ( \varphi^{[-1]}(x) +1) }{\varphi^\prime (1  ) } \right)^k dx
\end{equation}
\end{theorem}

Like with the Copula Graphon, the densities $t(C_3,W)$ and $t(P_i,W)$ for $i={2,3}$ can be further simplified for specific copula graphons as shown in Section \ref{c1a-numericalexperiments}. 

\subsection{Tensor Copula Graphon}\label{sec:c1a-tensorcopulagraphon}
Real world networks often have more heterogeneous mixing patterns and assortativity ranges than a single copula graphon provides \cite{han2022g_mixup, liu2024mixup}. A desired mixing pattern can be the result of a mixture of several copula graphons or copula density graphons. 

For a mixture of copula graphons or copula graphon densities $\hat{W}_1, \ldots \hat{W}_s$ the graphon tensor product $\mathbf{W}: [0,1]^s \rightarrow [0,1]$, is,
\begin{equation}
    \mathbf{W}(u_1,u_2, \ldots u_{2s-1}, u_{2s}) := \prod_{j=1}^s \hat{W}_j(u_{2j-1}, u_{2j})
\end{equation}

The tensor product graphon is not a graphon in the strictest sense \cite{lovasz2012large} but there exists a measure preserving map $\gamma ; [0,1] \rightarrow [0,1]^s$  :
\begin{align*}
     \mathbf{W}^{\gamma}(x_1,x_2, \ldots & x_{2s-1}, x_{2s}) =  \\ &\mathbf{W}(\gamma (x_1), \gamma (x_2), \ldots \gamma (x_{2s-1}), \gamma (x_{2s}))
\end{align*}
which enables generation from a tensor copula graphon in the same way as a copula graphon.

Let $G$ be a simple graph with the tensor graphon structure. The adjacency matrix $A$ of $G$ is as follows:
        \begin{equation}
            A_{ij}| \mathbf{u} \sim Bernoulli \left( \mathbf{W}(u_{2i-1}, u_{2j}\right) 
        \end{equation}
        where $u_i, u_j$ are uniform i.i.d random variables.
        
The tensor copula graphon allows for the control of the probabilities of motifs that determine $r$. For any graph $F$,
\begin{equation}
   t(F, \mathbf{W}) := \prod_{j=1}^s t(F,\hat{W}_j)
\end{equation}
Clearly, $t_{\mathrm{inj}}(F, \mathbf{W}) = t(F, \mathbf{W})$. The form of $r_\mathbf{W}$ is immediate.

\begin{theorem}[Degree Assortativity Coefficient with Tensor Copula Graphon Homomorphism Densities]\label{c1a-thm-tensorgraphonr}
Let $\mathbf{W}$ be a tensor copula graphon a mixture on $s$ copula graphons or copula graphon densities. The degree-degree assortativity coefficient characterized by homomorphism counts is:
\begin{equation}
    r_\mathbf{W} = \frac{(n-3) t(P_3,\mathbf{W}) + \frac{3n t(C_3, \mathbf{W})}{(n-1)(n-2)} - \frac{(n-2) t (P_2,\mathbf{W})^2}{t(P_1,\mathbf{W})}}{(n-3)3t(S_3,\mathbf{W}) + t(P_2,\mathbf{W}) -  \frac{(n-2) t(P_2,\mathbf{W})^2}{t(P_1,\mathbf{W})} }
\end{equation}
\end{theorem}

Proof follows immediately from (\ref{c1a-eq-tensorprod}) and Theorem \ref{c1a-thm-graphonr}.

\section{Methodology}
In this section, we present the algorithms that use Archimedean copulas to generate different network structures\footnote{The code for the implementation of both algorithms can be found in \url{https://github.com/videovic/GenCG}. }. We present copula graphon algorithm in which the Archimedean cumulative density function is used as the graphon. This algorithm can also be used for the copula density graphon. Also, we present the graphon tensor which is the tensor products of several copula graphons. We enable more flexible graph structures to be created to meet a target assortativity level $r$. In addition, we can control for the probabilities of homomorphism densities through $\theta$ in the different underlying copula graphons. 

\subsection{Estimating $\theta$}
Copulas dependence structures are fully characterized by $\theta$. In order to generate a network from any of the copula frameworks, estimation of $\theta$ for each copula or copula density is needed.

There are two immediate directions that can be used to determine $\theta$. First, is a method of moments style approach as graphons are fully characterized by the moments of their distribution, as presented in \cite{bickel2011method}. Suppose for a selected $W$, there is a subgraph density which is predefined, $p$. Then $\theta$ is the solution of, 
\begin{equation*}
    \min_\theta \left|t(F,W)-p \right|
\end{equation*}
From which this $\theta$ can be used to achieve target $r$. For example, estimating $\theta$ to generate a network to pre-specified $t(C,W)$ \cite{leung2007weighted}, which is the ratio of two homomorphism densities.

Alternatively, the estimation of $\theta$ through a numerical solver that minimizes $r$ is also possible. However, this approach is more complex to implement as the solver may face stability issues near the parameter boundary of the copulas. 

Note there are often many different possible combinations of $W$ and $\theta$ that can create the a target level of assortativity, the flexibility of these methods in reflecting the wide range of mixing patterns that create target assortativity in real world networks.

\subsection{Network Construction with the Copula Graphon} \label{sec-networkconstruction-copg}

We utilize the Archimedean copulas to generate graphons as their strength of dependence is determined by their parameter $\theta$. We have shown in Section \ref{ssec:copulagraphon} that homomorphism densities of the copula graphon are fully characterized by their generator function $\varphi_{\theta}$ and the strength of dependence by $\theta$. 

Synthetic graphs generated from the graphon will inherit properties from the underlying copula. Algorithm \ref{c1-algo-graphon} shows how to generate a copula graphon.

\begin{algorithm}[H] 
 \caption{Generative Pseudo Algorithm for a copula graphon or copula density graphon to desired $r$} 
 \label{c1-algo-graphon}
 \begin{algorithmic}[1]
    \renewcommand{\algorithmicrequire}{\textbf{Input:}}
    \renewcommand{\algorithmicensure}{\textbf{Return:}}
        \REQUIRE Function $F^\star$ which is bivariate Archimedean copula function $C$ or bivariate Archimedean copula density function $c$; parameter $\theta$; $n$ number of nodes
    \renewcommand{\algorithmicrequire}{\textbf{Init:}}       
    \REQUIRE Generate an empty adjacency matrix A with $n$ rows and $n$ columns
    \item[] \textbf{for each} $i=1$ to $j$ \textbf{do}
    \item[] Generate $A_{ij}$ which satisfies the following:
    \begin{equation}
        \mathbb{P}(A_{ij}=1 | u_1, u_2)= F^\star(u_i, u_j) 
    \end{equation}
     \item[] Replace lower triangular matrix with transpose of upper triangular matrix.
     \item[] Set diagonal= 0
     \textbf{end for} set matrix $A = (A_{ij})$
    \ENSURE Simple graph $W_G$ generated from $A$.
    \end{algorithmic}
\end{algorithm}

\subsection{Network Construction with Graphon Tensors} \label{sec-networkconstruction-copg-tensor}
Generating a tensor copula graphon requires specifying multiple copula structures to determine the probabilities of connection. Algorithm \ref{c1-algo-tensor} explains how to generate a tensor copula graphon. 

\begin{algorithm}[H] 
 \caption{Generative Pseudo Algorithm of the tensor copula graphon} \label{c1-algo-tensor}
 \begin{algorithmic}[1]
    \renewcommand{\algorithmicrequire}{\textbf{Input:}}
    \renewcommand{\algorithmicensure}{\textbf{Return:}}
        \REQUIRE A mixture $\mathbb{W}$ of bivariate Archimedean copula functions $\{ C_{s_1} \}$ and / or copula density functions $\{ c_{s_2} \}$ $s = s_1 + s_2$, $s > 1$, parameters $ \{\theta_j \}_{j=1}^s$, $n$ number of nodes
    \renewcommand{\algorithmicrequire}{\textbf{Init:}}       
    \REQUIRE Generate an empty adjacency matrix A with $n$ rows and $n$ columns
    \item[] \textbf{for each} $i=1$ to $j$ \textbf{do}
    \item[] Generate $2s$ i.i.d $u_i \sim U(0,1)$
    \item[] Generate $A_{ij}$ which satisfies the following:
    \begin{equation}
        \mathbb{P}(A_{ij}=1 | u_1, \ldots, u_{2s})= \prod_{t=1}^s \hat{W}_t(u_{2i-1}, u_{2j})
    \end{equation}
     \item[] Replace lower triangular matrix with transpose of upper triangular matrix.
     \item[] Set diagonal= 0
     \textbf{end for} set matrix $A = (A_{ij})$
    \ENSURE Simple graph $W_G$ generated from $A$.
    \end{algorithmic}
\end{algorithm}

\section{Numerical Experiments} \label{c1a-numericalexperiments}
To demonstrate the flexibility of the copula graphon and tensor graphon algorithms in Section \ref{sec-networkconstruction-copg} and Section \ref{sec-networkconstruction-copg-tensor}, examples are provided to generate a network to a target $r$. We also suggest what copula structures appropriate for target $r$. The role of $\theta$ in the resulting network assortativity is discussed. Each network is generated with $n=1000$ nodes and  $r$ is taken as the average of 10 simulations. Numerical simulations are carried out in $\mathrm{R}$.

\subsection{No assortativity, $r=0$}
The configuration model and preferential attachment produces random graphs with $r=0$. Likewise, this can also be achieved with the copula graphon by using $\Pi$ or copula graphons that will approximate $\Pi$ for some $\theta$ in their parameter space. Note that $\Pi$ does not induce any dependency structure between the nodes. For $W \sim \Pi$, $C \approx 0.04$ and thus $|P_{3/2}| \approx |P_{2/1}|$. 

Several copula graphons have parameter ranges within their parameter space or can approach the independence copula in their limit. For instance, the $W_J, W_G$ takes $\theta \in [1, \infty)$. $W_J, W_G$ produces graphs with $r=0$ at $\theta=1$, as $\theta \rightarrow \infty$ produces graphs with weak assortativity, $0< r < 0.5$. Also note for $W \sim W_C$, as $\theta \rightarrow 0$, $W_C \rightarrow \Pi$.

\begin{figure}[!t]
\centering
\includegraphics[width=2.5in]{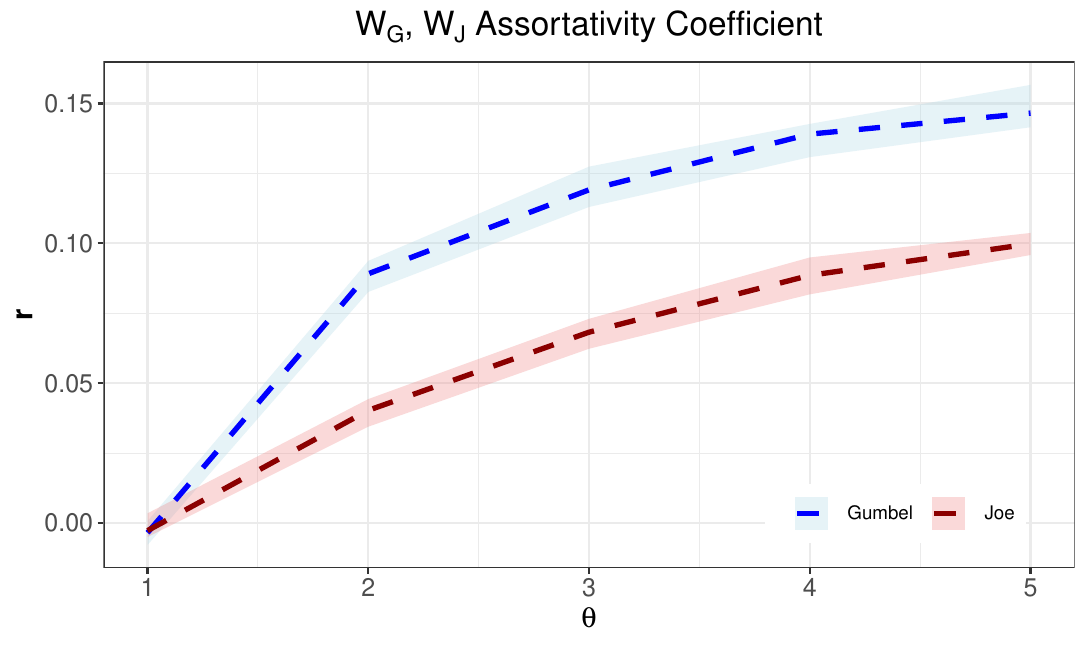}
\caption{Assortativity coefficient range by $W_G, W_J$ for $\theta \in [0,10]$. Blue, red line is the average assortativity coefficient of $W_G$, $W_J$ respectively, over 10 repetitions, the minimum and maximum within the range is the shaded region}
\label{zero-assort-01}
\end{figure}

The following experiment shows the rate of change of assortativity, once $\theta >1$ in $W_J$, $W_G$ where $1 < \theta_J, \theta_G < 5$. Figure \ref{zero-assort-01} verifies that assortativity coefficient starts at 0 but $W_G$'s assortativity rises at a faster rate. Although the clustering coefficient does not change much in this range. Increase of $W_G$'s assortativity is driven by an increasing but less variable $|P_3|$. 

\subsection{Assortative Networks}
Assortative networks have stronger mixing patterns with nodes of higher degree \cite{newman2002assortative}, which cannot be achieved by the copula graphons alone.  
Weakly assortative networks can be generated directly from the graphon in the copula framework. At the full range of $\theta$ in $W_G$'s parameter space, $W_G$ provides a full range of networks with $r$ from $0$ up to assortativity of $C^{+}$. As illustrated in Figure \ref{gumbel-weakassort-01}.
\begin{figure}[!t]
\centering
\includegraphics[width=2.5in]{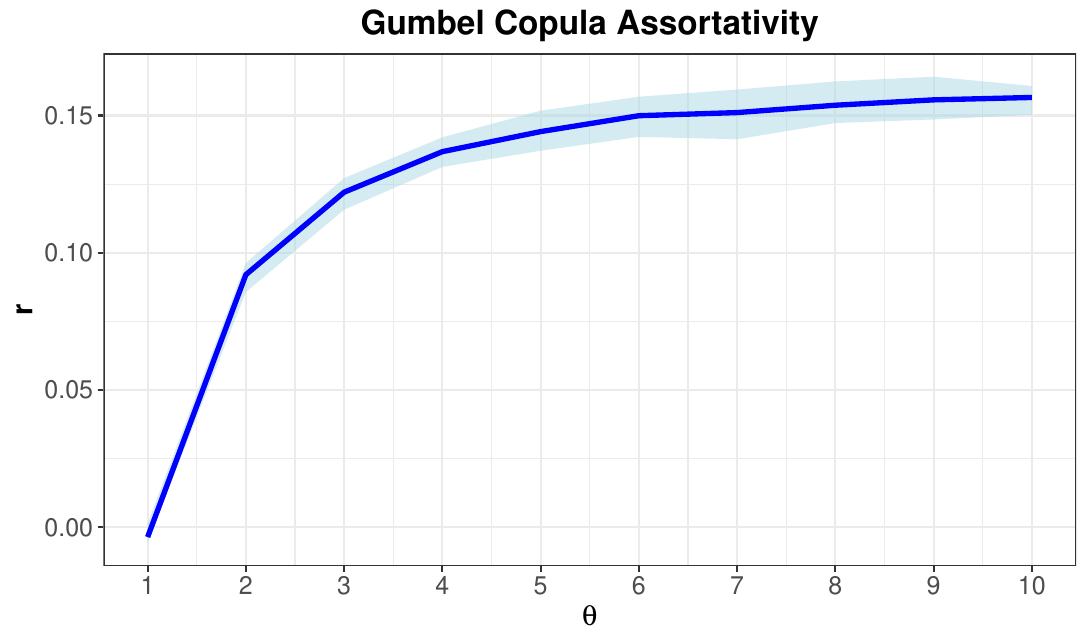}
\caption{Assortativity coefficient range by $W_G$ for $\theta \in [0,10]$. Blue line is the average assortativity coefficient over 10 repetitions, the minimum and maximum within the range is the shaded region}
\label{gumbel-weakassort-01}
\end{figure}

The densities of the copula graphons can produce a wider range of assortativity than the copula densities, as they are also joint probability density functions. The tensor product provides the flexibility to strengthen the connection probability when two moderately assortative copula graphons are combined.

Large values of $\theta$ have more pronounced dependency structure of the generated graphon. Using the tensor graphon, Figure \ref{gumbeljoe-highassort-01} illustrates $r \in (0.2,0.8)$ with $\Tilde{W}_J \otimes \Tilde{W}_G$ for $\theta_J = 2$, $ \theta_G \in[1,10]$. 

\begin{figure}[!t]
\centering
\includegraphics[width=2.5in]{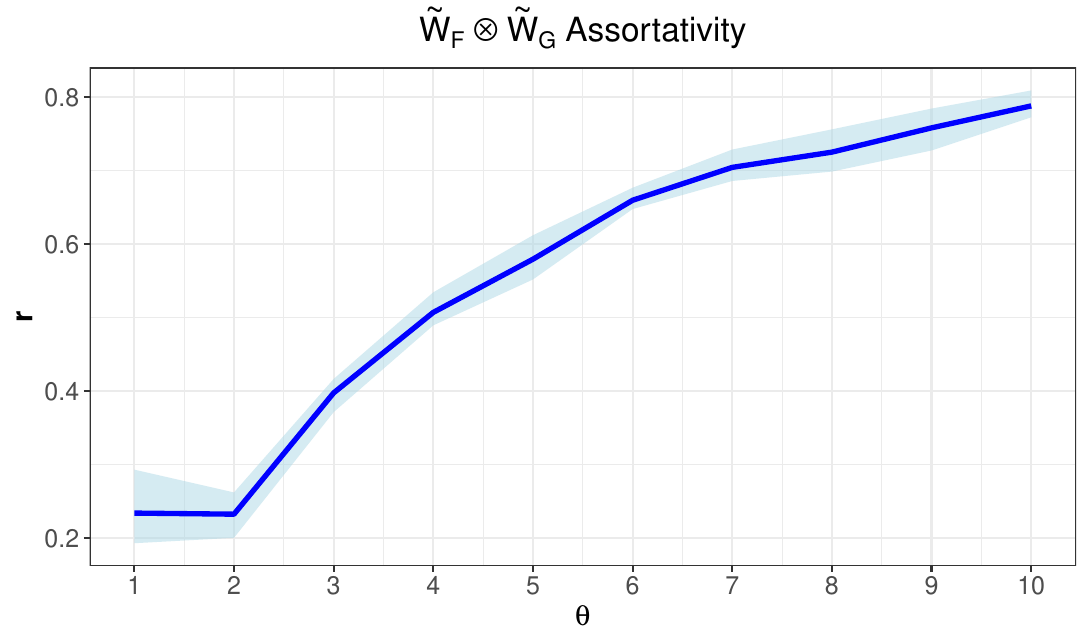}
\caption{Assortativity coefficient range through $\Tilde{W}_J \otimes \Tilde{W}_G$ for $\theta_J = 2$, $ \theta_G \in[1,10]$. Blue line is the average assortativity coefficient over 10 repetitions, the minimum and maximum within the range is the shaded region}
\label{gumbeljoe-highassort-01}
\end{figure}

\subsection{Disassortative Networks}
Disassortative networks can be generated by tensor product of weakly disassortative copula graphons. The tensor product can create more intricate/pronounced mixing patterns. We examine the $r$ of the tensor product of $W_F$ at the same $\theta$ value for $\theta_F \in [-10,-1]$, see Figure \ref{frank-lagged-disassort-02}. The tensor product produces $r = (-0.3,-0.5)$.
\begin{figure}[!t]
\centering
\includegraphics[width=2.5in]{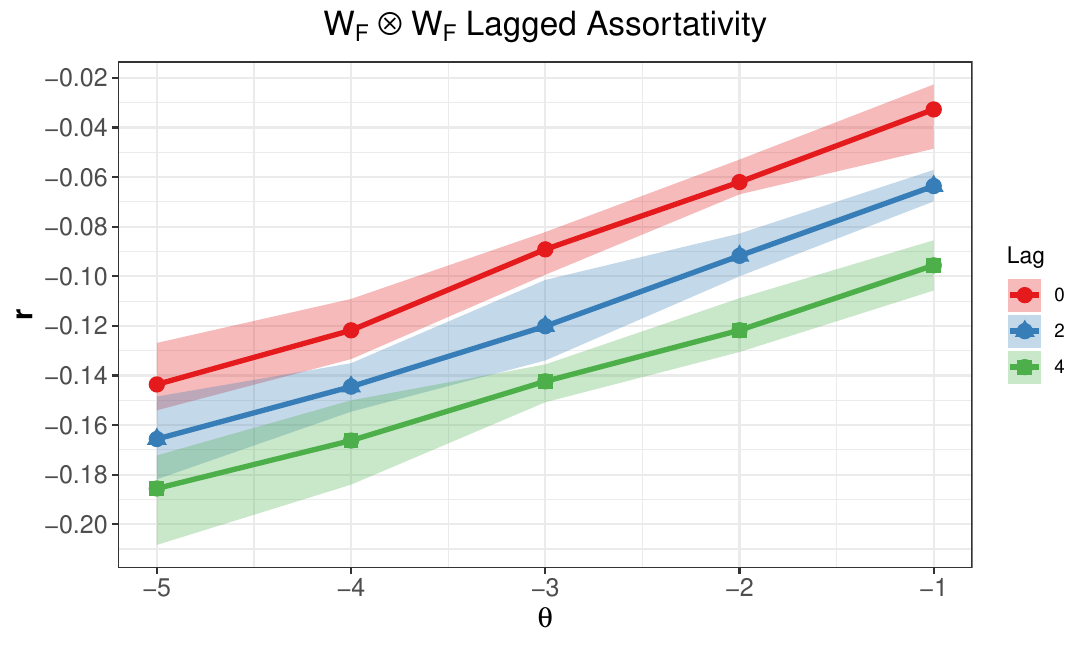}
\caption{Assortativity coefficient range through $\Tilde{W}_F \otimes \Tilde{W}_F$ for $\theta_F \in[1,5]$ at different lags. Assortativity coefficient is averaged over 10 repetitions. Line is the average assortativity coefficient over 10 repetitions, the minimum and maximum within the range is the shaded region}
\label{frank-lagged-disassort-02}
\end{figure}

The extent of disassortativity can be controlled by lagging $\theta$. By increasing the lag, the resultant disassortativity is lower but changes at a similar rate for increasing $\theta$. Increasing the lag in intervals of 2, creates almost parallel movement in the disassortativity range.

\section{Conclusion and Future Work}
This paper developed a method using copulas and graphons to generate a network to meet a target level of assortativity without using rewiring. First, we showed how the degree assortativity coefficient can be rewritten in terms of the homomorphism densities of a graphon and a simple random graph. Then, we showed how a graphon can be represented within a copula framework. Next, we provided three different ways to generate a network with copulas and provided generative algorithms with numerical examples. To validate our proposed model, we conducted extensive experiments to explore the possible assortativity different within the parameter space of selected copula graphons, copula density graphons and their tensor products. In the future, we will explore the possibility of using non-Archimedean copula structures to generate networks to target assortativity.

\end{document}